\documentclass[a4paper,11pt]{article}

\usepackage{multirow}

\usepackage{amsthm,amsfonts,amssymb,amsmath}
\usepackage{color}

\newenvironment{definition}[1][Definition]{\begin{trivlist} \item[\hskip \labelsep {\bfseries #1}]}{\end{trivlist}}
\newtheorem{theorem}{Theorem}
\newtheorem{proposition}{Proprosition}
\newtheorem{procedure}{Procedure}
\newtheorem{lemma}{Lemma}
\newtheorem{corollary}{Corollary}

\newcommand{\dirac}[1]{| #1 \rangle}

\begin{document}


\title{Non-Pauli Observables for CWS Codes}


\author{Douglas F.G. Santiago$^\dag$,
        Renato Portugal$^\ddag$,
        Nolmar Melo$^\ddag$ \\ \\
        \small $\dag$  Universidade Federal dos Vales do Jequitinhonha e Mucuri\\
 	\small  Diamantina, MG 39100000, Brazil \\
 	\small  douglassant@gmail.com \\
 	\small $\ddag$  Laborat\'{o}rio Nacional de Computa\c{c}\~{a}o Cient\'{\i}fica\\
 	\small      Petr\'{o}polis, RJ 25651-075, Brazil \\
  	\small      portugal@lncc.br, nolmar@lncc.br
 }





\maketitle


\begin{abstract}
It is known that nonadditive quantum codes are more optimal for error correction when compared to stabilizer codes. The class of codeword stabilized codes (CWS) provides tools to obtain new nonadditive quantum codes by reducing the problem to finding nonlinear classical codes. In this work, we establish some results on the kind of non-Pauli operators that can be used as decoding observables for CWS codes and describe a procedure to obtain these observables.
\end{abstract}

\section{Introduction}	

It is known that quantum computers are able to solve hard problems in polynomial time and to increase the speed of many algorithms~\cite{Shor1994,Grover1996,Mosca2009,Childs2010}. Decoherence problems are present in any practical implementation of quantum devices, especially in large-scale quantum computer. Quantum error correcting codes (QECCs) can be used to solve these problems by using extra qubits and storing information using redundancy~\cite{Calderbank1996,Bennett1996,Steane1996,Knill1997}.

The framework of stabilizer codes was used to obtain a large class of important quantum codes~\cite{Gottesman1996,Gottesman1997,Calderbank1997}. A code is called a stabilizer code if it is in the joint positive eigenspace of a commutative subgroup of Pauli group. In certain cases, these codes are suboptimal, because there is larger class, called nonadditive codes.

An important class of nonadditive codes, called CWS, has been studied recently~\cite{Smolin2007,Yu2008,Cross2008,Chen2008,Chuang2009}. The framework of CWS codes generalizes the stabilizer code formalism and has been used to build some good nonadditive codes, in some cases enlarging the logical space of stabilizer codes of the same length. On the one hand, many papers address codification procedures for CWS codes, and on the other hand few papers address decodification procedures. Decoding observables of specific codes are known, such as the ((9,12,3)) and ((10,24,3)) codes and associated families~\cite{Yu2008,Yu2009}. A generic decoding procedure for binary CWS codes was proposed in Ref.~\cite{Li2010} and extended for nonbinary CWS codes in Ref.~\cite{Melo2012}.

In this work, we establish a condition to the existence of non-Pauli CWS observables, that can be written in terms of the stabilizers associated to the CWS code. We describe a procedure to find these observables, which is specially useful for CWS codes that are close to stabilizer codes.


This paper is divided in the following parts. In Section~2, we review the structure of CWS codes and introduce the notations that will used in this work. In Section~3, we present the main results, in special Theorem~\ref{Principal}, its corollary and the procedure to find non-Pauli observables. In Section~4, we give an example and in Section~5, we present the conclusions.


\section{CWS codes}

An $((n,K))$ CWS code in the Hilbert space $\mathcal{H}^n$ is described by
\begin{enumerate}
\item A stabilizer group $S=\langle s_1,\ldots, s_n\rangle$, where $\{s_i\}$ is a generator set of independent and commutative Pauli operators  {(elements of Pauli group ${\mathcal G}_n$)}. This group stabilizes a single codeword $\dirac{\psi}$;
\item A set of Pauli operators $W=\{W_1,\ldots,W_K\}$. The set $\{W_j\dirac{\psi}\}$ spans the CWS code and each $W_i$ is called a  {codeword} operator.
\end{enumerate}

Cross \textit{et al.}~\cite{Cross2009} have showed that any binary CWS code is equivalent to a CWS code in a standard form, which is characterized by: (1) a graph of $n$ vertices, (2) a set of Pauli operators $s_i=X_iZ^{r_i}$, where $r_i$ is the $i$-th line of the adjacency matrix $(M)$, and (3) codeword operators $W_j=Z^{C_j}$, where $C_1=(0,\ldots,0)$, that is, $W_1=I$.

In the standard form, correctable Pauli errors can be expressed as binary strings. A Pauli error $E=Z^VX^U$, where $V,U\in \mathbb{F}_2^n$, can be mapped modulo a phase to an error $Z^{\mathrm{Cl}_S(E)}$ through function
$$ {\mathrm{Cl}_S(Z^VX^U)=V+MU \in \mathbb{F}_2^n.}$$
The problem of finding good CWS quantum codes is reduced to the problem of finding good classical codes.  {Theorem 3 of Cross \textit{et al.}~\cite{Cross2009} states that a CWS code in standard form with stabilizer $S$ spanned by $\{Z^{C_i}\dirac{\psi}\}$ detects errors in the set ${\mathcal{E}}=\{E_i\}$ if and only if the classical code $\{C_i\}$ detects errors in $\mathrm{Cl}_S({\mathcal{E}})$.} This result is valid because, for all $E_i$ satisfying $\mathrm{Cl}_S(E_i)=0$, we disregard all binary vectors $C$ such that $Z^{C}E_i=E_iZ^{C}$.

 {Our first goal is to analyze which Pauli operators can be used as observables for CWS codes. If $W$ is the set of codeword operators and $g\in {N_S(W)}$, where ${N_S(W)}$ is the normalizer of $W$ in $S$, $g$ can be used as a decoding observable. This follows from the equalities}
\begin{align*}
gE_iW_j\dirac{\psi}=m_iE_igW_j\dirac{\psi}=m_iE_iW_jg\dirac{\psi}=m_iE_iW_j\dirac{\psi},
\end{align*}
where $m_i=\pm 1$. It means that, for a fixed $E_i$ and for all $W_j$, $E_iW_j\dirac{\psi}$ lies entirely in the eigenspace associated with the  {eigenvalue} $m_i$ of $g$.  {So, there is no information leakage after the measurement of observable $g$. When a CWS code is a stabilizer code, the decoding procedure uses a generating set of ${N_S(W)}$ as observables. This is not the only choice, because we can use non-Pauli observables.}

 {Our second goal is to} establish some results on the existence and form of non-Pauli CWS observables on the group algebra $\mathbb{R}[S]$ over $\mathbb{R}$ spanned by $S$. An operator $A\in \mathbb{R}[S]$ can be written as
$$\displaystyle A=\sum_{V\in \mathbb{F}_2^n}\alpha_{V}{\cal S}^{V},$$
 {where we use the notation ${\cal S}^{V}$ as an element of $S=\langle s_1,\ldots, s_n\rangle$ given by}
$${\cal S}^{V}=s_1^{v_1} {\cdots }s_n^{v_n},$$
where $V=(v_i,\ldots,v_n)$ is a binary vector.  {We will assign a type to operator $A$ depending on the number of non-zeros coefficients $\alpha_V$. This type notion is captured in the next definition.}

\begin{definition}
A type-$i$ observable is an operator $A\in \mathbb{R}[S]$ that satisfies $A^2=I$ and is exactly a linear combination of $i$ different elements of $S$.
\end{definition}

 {
Note that this definition makes sense because group $S$ is a subset of a basis of the Hilbert space, and $\displaystyle A=\sum_{V\in \mathbb{F}_2^n}\alpha_{V}{\cal S}^{V}$ is written in a unique way.
}

Type-$1$ observables are Pauli operators.  {It is straightforward to show that there are no type-2 or type-3 observables.} In this work, we consider only type-4 observables. 

\section{Main Results}

 {If a unitary operator $A$ is an observable, then $A^2=I$. Since we are dealing with observables in $\mathbb{R}[S]$, we have the following proposition:}

\begin{proposition}
 Let $S$ be the stabilizer group of a CWS code in standard form and $\displaystyle A=\sum_{V\in \mathbb{F}_2^n}\alpha_{V}{\cal S}^{V}$ an element of $\mathbb{R}[S]$. Then, $A^2=I$ if and only if

 \begin{equation}\label{eqgeral}
      \sum_{V\in \mathbb{F}_2^n}\alpha_{V}^2=1 \textrm{{ and }}  {                    \sum_{V\in \mathbb{F}_2^n } \alpha_V \alpha_{V+U}=0,}\; \forall U \in \mathbb{F}_2^n\setminus\{0\}.
\end{equation}
\end{proposition}

\begin{proof}
Take
\[
         A^2 = \sum_{V\in \mathbb{F}_2^n }\alpha_{V}^2 I + \sum_{V\neq V'}\alpha_{V}\alpha_{V'} {\cal S}^V {\cal S}^{V'}.
\]
All terms ${\cal S}^U\in S\setminus\{I\}$ are present in the second sum, each one {as many times as} $V+V'=U$, {that is, $2^n$.} So, we can rewrite this equation as
 {
\begin{eqnarray*}
         A^2  & = & \displaystyle \sum_{V\in \mathbb{F}_2^n}\alpha_{V}^2 I + \displaystyle \sum_{U\in \mathbb{F}_2^n\setminus \{0\}}  \sum_{V+V' =U} \alpha_V \alpha_{V'} {\cal S}^U\\
         & = & \displaystyle \sum_{V\in \mathbb{F}_2^n}\alpha_{V}^2 I + \displaystyle \sum_{U\in \mathbb{F}_2^n\setminus \{0\}} {\cal S}^U \sum_{V\in \mathbb{F}_2^n} \alpha_V \alpha_{V+U}.
\end{eqnarray*}
}
Then, result (\ref{eqgeral}) follows.

\end{proof}

 {Type-4 observables can be restricted by the following theorem:}

\begin{theorem}\label{teor:type4}
\textit{A is a type-4 observable if and only if
\begin{equation}\label{eqAtype4}
    A= \pm \frac{{\cal S}^V}{2}\left( -I+{\cal S}^{V_1}+{\cal S}^{V_2}+{\cal S}^{V_1+V_2}\right)
\end{equation}
  {with $V_1\neq V_2\in \mathbb{F}_2^n\setminus \{0\}$ and $V\in \mathbb{F}_2^n$.}}
\end{theorem}

\begin{proof}
If $A$ is given by Eq.~(\ref{eqAtype4}), then it is straightforward to verify that $A^2=I$. So, $A$ is a type-4 observable.

Reciprocally, take a type-4 observable $A=\alpha_1{\cal S}^{U_1}+\alpha_2{\cal S}^{U_2}+\alpha_3{\cal S}^{U_3}+\alpha_4{\cal S}^{U_4}$. We have
\begin{eqnarray*}
  A^2 &=& \left( \sum_{i=1}^{4}\alpha_i^2\right)I+2\alpha_1 \alpha_2 {\cal S}^{U_1+U_2}+2\alpha_1 \alpha_3 {\cal S}^{U_1+U_3}+  2\alpha_1\alpha_4 {\cal S}^{U_1+U_4}+ \\
   && 2\alpha_2 \alpha_3 {\cal S}^{U_2+U_3} +2\alpha_2\alpha_4 {\cal S}^{U_2+U_4}+2\alpha_3 \alpha_4 {\cal S}^{U_3+U_4}.
\end{eqnarray*}
 {The $\alpha$'s are not zero. So, $A^2=I$ implies that
$$\sum_{i=1}^{4}\alpha_i^2=1$$
and the sum of the remaining 6 terms is zero, which implies that $U_1+U_2=U_3+U_4$, $U_1+U_3=U_2+U_4$ and $U_1+U_4=U_2+U_3$.}

We can rewrite $A$ by taking  {$V=U_1$, $V_1=U_1+U_2$ and $V_2=U_1+U_3$, then $V_1+V_2 = U_1+U_4$ and}
\[
 A=  \frac{{\cal S}^{V}}{2}\left(\alpha_1 I+\alpha_2{\cal S}^{V_1}+\alpha_3 {\cal S}^{V_2}+\alpha_4 {\cal S}^{V_1+V_2}\right).
\]
Note that $V_1\neq V_2$ and $V_1\neq 0 \neq V_2$ because $U_i\neq U_j$, if $i\neq j$.
The solutions obeying constraints \eqref{eqgeral} belong to the set
\begin{eqnarray*}
         (\alpha_1,\alpha_2,\alpha_3,\alpha_4) \in &
	      \pm\frac 12\{ & (-1,1,1,1),(1,-1,1,1),(1,1,-1,1),(1,1,1,-1)\}.  
\end{eqnarray*}
The last three solutions can be obtained from the first one by  {collecting} ${\cal S}^{V_1}$, ${\cal S}^{V_2}$ and ${\cal S}^{V_1+V_2}$, respectively, and absorbing in ${\cal S}^V$.

\end{proof}

Let us introduce the following notation:
\begin{eqnarray}
	  {\cal S}^{(V_1,V_2)}=\frac{1}{2}\left( -I+{\cal S}^{V_1}+{\cal S}^{V_2}+{\cal S}^{V_1+V_2}\right).
\end{eqnarray}
Note that,  {for any $V_1,V_2\in \mathbb{F}_2^n$,} ${\cal S}^{(V_1,V_2)}$ stabilizes $\dirac{\psi}.$  {In the next Lemma,  we use function $F:{\mathcal G_n} \mapsto\mathbb{F}_2^n$, which depends implicitly on $V_1$ and $V_2$, and is defined by
\begin{align}\label{commut} F(G)=\left\{\begin{array}{cl}
V_1+V_2 & \textrm{if }G\textrm{ anticommute with } {\cal S}^{V_1} \textrm{ and } {\cal S}^{V_2};   \\
V_1 & \textrm{if } G \textrm{ anticommute only with }{\cal S}^{V_2};\\
V_2 & \textrm{if } G \textrm{ anticommute only with }{\cal S}^{V_1};\\
 0   & \textrm{otherwise. }
\end{array}\right.\end{align}
}

\begin{lemma}\label{lema1}
Let $G$ be a Pauli operator. If $G$ does not commute with ${\cal S}^{V_1}$ or ${\cal S}^{V_2}$,  then ${\cal S}^{(V_1,V_2)}G=-{\cal S}^{F(G)}G\,{\cal S}^{(V_1,V_2)}=-G\,{\cal S}^{F(G)}{\cal S}^{(V_1,V_2)}=-G\,{\cal S}^{(V_1,V0_2)}{\cal S}^{F(G)}$.
 \end{lemma}

\begin{proof}
 {The verification is straightforward.}
\end{proof}

The conditions to use an operator $A$ as a CWS observable is closely related to the conditions that guarantees that $A$ stabilizes the code.

\begin{proposition}\label{auxPrincipal}
 Let $C_i=(c_i^1,\ldots,c_i^n)$, $i=1,\ldots, K$ be the classical codewords of a CWS code in standard form. Let $V_1$, $V_2$, $V\in \mathbb{F}_2^n$ and $p_{i}=\langle C_i,V_1\rangle \vee \langle C_i,V_2\rangle$. Then, a type-4 observable $A$ stabilizes the code if and only if $A={\cal S}^{V}{\cal S}^{(V_1,V_2)}$ and $\langle C_i,V\rangle=p_i$ for all $i$.
\end{proposition}

\begin{proof}
Suppose that $A={\cal S}^{V}{\cal S}^{(V_1,V_2)}$ and $\langle C_i,V\rangle=p_i$ is true for all $i$. Then,
\begin{enumerate}
\item if $Z^{C_i}$ commutes with ${\cal S}^{V_1}$ and ${\cal S}^{V_2}$, then $Z^{C_i}$ also commutes with ${\cal S}^{V}$ and ${\cal S}^{(V_1,V_2)}$, that is,
$${\cal S}^{V}{\cal S}^{(V_1,V_2)}Z^{C_i}\dirac{\psi}=Z^{C_i}\dirac{\psi},$$
\item if $Z^{C_i}$ does not commute with ${\cal S}^{V_1}$ or ${\cal S}^{V_2}$, then $Z^{C_i}$ anticommutes with ${\cal S}^{V}$.  {Besides, Lemma~\ref{lema1} implies that}
\begin{align*}&{\cal S}^{V}{\cal S}^{(V_1,V_2)}Z^{C_i}\dirac{\psi}={\cal S}^{V}(-{\cal S}^{ F(Z^{C_i})})Z^{C_i}{\cal S}^{(V_1,V_2)}\dirac{\psi}=-{\cal S}^{V}Z^{C_i}{\cal S}^{ F(Z^{C_i})}\dirac{\psi}=\\
&-{\cal S}^{V}Z^{C_i}\dirac{\psi}=Z^{C_i}{\cal S}^{V}\dirac{\psi}=Z^{C_i}\dirac{\psi}.\end{align*}
\end{enumerate}
 {In all cases, $A$ stabilizes the code.}

Reciprocally, let \(A\) be a type-4 observable. By Theorem~\ref{teor:type4}, we have \(A=\pm {\cal{S}}^V{\cal{S}}^{\{V_1,V_2\}}\). Then, \(A\dirac{\psi} = \pm {\cal{S}}^V{\cal{S}}^{\{V_1,V_2\}}\dirac{\psi} = \pm \dirac{\psi} \). By supposition, \(A\) stabilizes the code. Therefore, \(A= {\cal{S}}^V{\cal{S}}^{\{V_1,V_2\}}\). Besides, to stabilize the code, we have:

\begin{enumerate}
\item If  {a codeword operator $W_i=Z^{C_i}$} commutes with ${\cal S}^{V_1}$ and ${\cal S}^{V_2}$, we have
$${\cal S}^{V}{\cal S}^{(V_1,V_2)}Z^{C_i}\dirac{\psi}={\cal S}^{V}Z^{C_i}{\cal S}^{(V_1,V_2)}\dirac{\psi}={\cal S}^{V}Z^{C_i}\dirac{\psi}=Z^{C_i}\dirac{\psi}.$$
 {The last equality implies that ${\cal S}^{V}$ commutes with $Z^{C_i}$. So,} $\langle C_i,V\rangle=0$.
\item If $W_i=Z^{C_i}$ does not commute with ${\cal S}^{V_1}$ or ${\cal S}^{V_2}$, the Lemma~\ref{lema1} implies that
\begin{align*}&{\cal S}^{V}{\cal S}^{(V_1,V_2)}Z^{C_i}\dirac{\psi}={\cal S}^{V}(-{\cal S}^{ F(Z^{C_i})})Z^{C_i}{\cal S}^{(V_1,V_2)}\dirac{\psi}=-{\cal S}^{V}Z^{C_i}{\cal S}^{ F(Z^{C_i})}\dirac{\psi}\\
&=-{\cal S}^{V}Z^{C_i}\dirac{\psi}=Z^{C_i}\dirac{\psi}.
\end{align*}
 {The last equality implies that ${\cal S}^{V}$ anticommutes with $Z^{C_i}$. So,} $\langle C_i,V\rangle=1$.
\end{enumerate}
 {These results show that $\langle C_i,V\rangle=p_i$ is true for all $i$.}
\end{proof}

Taking $V=(v_1,\ldots, v_n)\in \mathbb{F}_2^n$ and  $p_{i}=\langle C_i,V_1\rangle \vee \langle C_i,V_2\rangle$, equations $\langle C_i,V\rangle=p_i$ can be put in matrix form
\begin{align}\label{slr}
C
\left[\begin{array}{c}
         v_1\\\vdots\\v_n
        \end{array}
\right]
=
\left[\begin{array}{c}
         p_1\\\vdots\\p_k
        \end{array}
\right],
\end{align}
where $C$ is the matrix of all classical codewords
\begin{equation}\label{matrixC}
  C=\left[\begin{array}{ccc}
         c_1^1&\ldots &c_1^n\\\vdots&\vdots&\vdots\\ c_K^1& \ldots &c_K^n
        \end{array}
\right].
\end{equation}

An operator $A$ can be used as a CWS observable in the decoding procedure, if the encoded information is not lost after the measurement of $A$. We have to guarantee that, for each $i$ and for all~$j$, $E_iW_j\dirac{\psi}$ belongs to the eigenspace of $E_iW_j$ associated with the eigenvalues 1 or -1, that is, $$AE_iW_j\dirac{\psi}=E_iW_j\dirac{\psi},\;\forall j$$
or
$$AE_iW_j\dirac{\psi}=-E_iW_j\dirac{\psi},\;\forall j.$$
Those facts lead us to the following theorem:

\begin{theorem}\label{Principal}
Let $\mathcal{E}=\{E_i\}_{i=1}^{T}$ be a set of correctable Pauli errors of a CWS code in standard form. Then, a type-4 observable $A= {\cal S}^{V}{\cal S}^{(V_1,V_2)}$ can be used as  {a decoding} observable if and only if  for all $i\in\{1,\ldots,T\}$ there is $V'_i$ solution of Eq.~\eqref{slr} with $V=V'_i+F(E_i)$.
\end{theorem}

\begin{proof}
Suppose that $A={\cal S}^{V}{\cal S}^{(V_1,V_2)}$ satisfies $V=V_i+F(E_i)$, for all $i$, where $V_i$ is a solution of Eq.~\eqref{slr}. Let ${\cal S}^{V}E_i=m_iE_i{\cal S}^{V}$, where $m_i=\pm 1$. Then,  {by Lemma \ref{lema1} we have}
\begin{enumerate}
\item if $E_i$ commutes with ${\cal S}^{V_1}$ and ${\cal S}^{V_2}$, then $F(E_i)=(0,\ldots, 0)$~(\ref{commut}) and
 {
\begin{eqnarray*}
	AE_iW_j\dirac{\psi}& = &  {\cal S}^{V}{\cal S}^{(V_1,V_2)}E_iW_j\dirac{\psi}= {\cal S}^{V}E_i{\cal S}^{(V_1,V_2)}W_j\dirac{\psi}\\
	& = & m_iE_i{\cal S}^{V}{\cal S}^{(V_1,V_2)}W_j\dirac{\psi}= m_iE_i{\cal S}^{V+F(E_i)}{\cal S}^{(V_1,V_2)}W_j\dirac{\psi}\\
	& = & m_iE_i{\cal S}^{V'_i}{\cal S}^{(V_1,V_2)}W_j\dirac{\psi}= m_iE_iW_j\dirac{\psi}.
\end{eqnarray*}}
The last equality holds because  {${\cal S}^{V'_i}{\cal S}^{(V_1,V_2)}$} stabilizes the code.

\item If $E_i$  {does} not commute with ${\cal S}^{V_1}$ or ${\cal S}^{V_2}$, then
 {
\begin{eqnarray*}
	AE_iW_j\dirac{\psi} & = &  {\cal S}^{V}{\cal S}^{(V_1,V_2)}E_iW_j\dirac{\psi}= -  {\cal S}^{V+F(E_i)}E_i{\cal S}^{(V_1,V_2)}W_j\dirac{\psi}\\
	& = & -  m_iE_i{\cal S}^{V'_i}{\cal S}^{(V_1,V_2)}W_j\dirac{\psi} = -  m_iE_iW_j\dirac{\psi}.
\end{eqnarray*}}
Again, we have used that  {${\cal S}^{V'_i}{\cal S}^{(V_1,V_2)}$} stabilizes the code.
\end{enumerate}

Reciprocally, suppose that $A= {\cal S}^{V}{\cal S}^{(V_1,V_2)}$ can be used as a decoding CWS observable. Using ${\cal S}^{V}E_i=m_iE_i{\cal S}^{V}$, where $m_i=\pm 1$, and repeating the commuting process, we have
\begin{enumerate}
\item if $E_i$ commutes with both ${\cal S}^{V_1}$ and ${\cal S}^{V_2}$, then
\begin{align*}
&AE_iW_j\dirac{\psi}= {\cal S}^{V}{\cal S}^{(V_1,V_2)}E_iW_j\dirac{\psi}=
m_iE_i{\cal S}^{V+F(E_i)}{\cal S}^{(V_1,V_2)}W_j\dirac{\psi}
\end{align*}
where $F(E_i)=(0,\ldots,0)$.
\item If $E_i$  {does} not commute with ${\cal S}^{V_1}$ or ${\cal S}^{V_2}$, then
\begin{align*}
&AE_iW_j\dirac{\psi}= {\cal S}^{V}{\cal S}^{(V_1,V_2)}E_iW_j\dirac{\psi} = -  m_i
E_i{\cal S}^{V+F(E_i)}{\cal S}^{(V_1,V_2)}W_j\dirac{\psi}.
\end{align*}
\end{enumerate}
We are assuming that $A$ can be used as a decoding CWS observable. In both cases, we have
\begin{align*}
AE_iW_j\dirac{\psi}=E_iW_j\dirac{\psi}, \;\forall j
\end{align*}
or
\begin{align*}
AE_iW_j\dirac{\psi}=-E_iW_j\dirac{\psi}, \;\forall j.
\end{align*}
This implies that ${\cal S}^{V+F(E_i)}{\cal S}^{(V_1,V_2)}$ stabilizes the code for all $i$, and by Prop.~\ref{auxPrincipal} there is a solution  {$V'_i$} of Eq.~\eqref{slr} such that  {$V+F(E_i)=V'_i$ for all $i$.}

\end{proof}

Theo.~\ref{Principal} allows us to make an exhaustive search for type-4  {decoding} observables using expression $A= {\cal S}^{V}{\cal S}^{(V_1,V_2)}$. We have to consider all pairs $({\cal S}^{V_1},{\cal S}^{V_2})$ in $S$ such that $V_1\neq V_2$ and look for solutions of Eq.~(\ref{slr}) for each pair.  {This process can be expensive. Next corollary addresses a more efficient way to search the decoding observables by restricting the search space to ${N_S({\mathcal{E}})}$. In this case, some solutions may be lost.}

\begin{corollary}\label{corollary1}
Let $\mathcal{E}=\{E_i\}_{i=1}^{T}$ be a set of correctable errors of a CWS code in standard form and ${N_S({\mathcal{E}})}$ the normalizer of $\mathcal{E}$ in $S$. If $A= {\cal S}^{V}{\cal S}^{(V_1,V_2)}$ is a type-4 observable, where ${\cal S}^{V_1},{\cal S}^{V_2}\in {N_S({\mathcal{E}})}$ and $V$ is a solution of Eq.~\eqref{slr}, then $A$ is a  {decoding observable for the CWS code.}
\end{corollary}

\begin{proof}
If both ${\cal S}^{V_1}$ and ${\cal S}^{V_2}$ are in ${N_S({\mathcal{E}})}$, then $F(E_i)=(0,\ldots,0)$ for all $i$, and Theo.~\ref{Principal} implies that $V=V_i$, where $V_i$ is a solution of Eq.~\eqref{slr}.

\end{proof}

Corollary~\ref{corollary1} helps us to build a procedure to find type-4 decoding observables, which we describe now.

\begin{procedure}\label{procedure1}
Let $\mathcal{E}=\{E_i\}$ be the set of correctable errors and $W=\{W_j\}$ the set of codeword operators.
\begin{enumerate}
\item Find independent generators of ${N_S(W)}$.
\item Measure the generators. For each sequence of measurement results, there is set ${\mathcal{E}'}$, subset of ${\mathcal{E}}$, of errors that were not detected by the measurements.
\item For each ${\mathcal{E}'}$ do
\begin{enumerate}
  \item Find all elements in group ${N_S({\mathcal{E}'})}$.
  \item Take pairs $({\cal S}^{V_1}, {\cal S}^{V_2})$ in ${N_S({\mathcal{E}'})}$ such that ${V_1}\neq {V_2}$  until finding a solution $V$ of Eq.~(\ref{slr}) that distinguishes some errors in ${\mathcal{E'}}$. This step may split ${\mathcal{E}'}$ into smaller subsets.
  \item Repeat Step (a) and (b) with smaller subsets as many times as needed until  {distinguishing Pauli errors in ${\mathcal{E}'}$.}
\end{enumerate}
\end{enumerate}
\end{procedure}

To find generators of ${N_S(W)}$ in Step~1, we employ the commuting relations
\begin{equation}\label{}
  Z^{C_i} \mathcal{S}^{O_j} = (-1)^{\langle C_i, O_j \rangle}   \mathcal{S}^{O_j} Z^{C_i}
\end{equation}
to show that $\mathcal{S}^{O_j}\in {N_S(W)}$ if and only if $\langle C_i, O_j \rangle=0$ for all $i$. This implies that $O_j$ must be in the kernel of matrix $C$, described in Eq.~\eqref{matrixC}. The independent generators for ${N_S(W)}$ are obtained from a basis of the kernel of $C$.

To find all elements in ${N_S(\mathcal{E}')}$ in Step~3(a), we convert the errors in $\mathcal{E}'$ to classical words by using function Cl$_S$ and build a new matrix. The kernel of this matrix is in one-to-one correspondence to the elements of ${N_S({\mathcal{E}'})}$. Each pair $({\cal S}^{V_1}, {\cal S}^{V_2})$ and a solution $V$ of Eq.~(\ref{slr}) provides a non-Pauli observable for errors in ${\mathcal{E'}}$. Step~3 can be improved by testing whether each non-Pauli observable can be used for other sets  ${\mathcal{E'}}$ generated is Step~2.

\section{Example}

In this section,  we employ Procedure~\ref{procedure1} to find the decoding observables for the $((10,20,3))$ code, described by Cross \textit{et al.}~\cite{Cross2009}. This code is based on the double ring graph, with the following generators:
$$\begin{array}{cc}
s_1=XZIIZZIIII & \,\,\,\,\,s_6=ZIIIIXZIIZ \\
s_2=ZXZIIIZIII & \,\,\,\,\,s_7=IZIIIZXZII \\
s_3=IZXZIIIZII & \,\,\,\,\,s_8=IIZIIIZXZI \\
s_4=IIZXZIIIZI & \,\,\,\,\,s_9=IIIZIIIZXZ \\
s_5=ZIIZXIIIIZ & \,\,\,\,\,s_{10}=IIIIZZIIZX \\
\end{array}$$
The associated classical codewords are
$$\begin{tabular}{cccc}
0000000000 & 1001100100 & 1001101111 & 0101100000\\
0000101001 & 1100101101 & 0111011011 & 0111010000\\
1011011111 & 1110010110 & 1100000100 & 1101111110\\
1111000101 & 0101101011 & 0001111010 & 0010010010\\
0010111011 & 1011010100 & 0011000001 & 1110111111
\end{tabular}$$

In Step~1 of Procedure~\ref{procedure1}, we have to find generators for ${N_S(W)}$. This is accomplished by finding a basis $(O)$ for the kernel of matrix $C$, described in Eq.~\eqref{matrixC}. In this example, this basis is given in Table~\ref{table10}. Then, the generators of ${N_S(W)}$ are Pauli observables ${\cal S}^{O_1},$ ${\cal S}^{O_2}$, ${\cal S}^{O_3}$, ${\cal S}^{O_4}$. In Step~2, they are measured one at a time. The results are displayed as signs $\pm$ on the top of subtables in Fig.~\ref{table30}. For example, if the results of measuring these Pauli observables are $+++ -$, only two Pauli errors were not detected, namely, $Y_2$ and $Z_1$. ${\mathcal{E}'}$ is $\{Y_2,Z_1\}$ in this case.

\begin{table}[h!]
  \centering
  \caption{Decoding observables (Pauli type --- ${\cal S}^{O_i}$) for the ((10,20,3)) code.}\label{table10}
\begin{tabular}{|c|c|c|c|} \hline
   $O_{1}$ & $O_{2}$ & $O_{3}$ & $O_{4}$ \\ \hline
0001110011 &  0010011001 & 0100111110 & 1000000100\\ \hline
\end{tabular}
\end{table}

In Step~3 of Procedure~\ref{procedure1}, we obtain the first non-Pauli observable, $A_1$ in Table~\ref{table20}, when ${\mathcal{E}'}=\{Y_2,Z_1\}$ . In this case, Step~3(a) is used only one time, because observable $A_1$ distinguishes all errors in ${\mathcal{E}'}$. Note that we can verify whether $A_1$ can be used for others ${\mathcal{E}'}$. In this example, $A_1$ can be used 4 times, as can be seen in Fig.~\ref{table30}. The next set will be ${\mathcal{E}'}=\{X_4,Z_3\}$.

\begin{table}[h!]
  \centering
  \caption{Decoding observables (non-Pauli) for the ((10,20,3)) code. They are type-4 observables described by $A={\cal S}^{V}{\cal S}^{(V_1,V_2)}$ (see Eq.~(\ref{eqAtype4})).}\label{table20}
\begin{tabular}{|c|c|c|c|}
\hline
        & $V$           &   $V_1$        &   $V_2$ \\\hline
$A_1$   & 0000111001   &  0000100001 & 0001000011 \\\hline
$A_2$   & 0000111001   &  0000100010 & 0001000000  \\\hline
$A_3$   & 0000010001   &  0000000011 & 0000010010  \\\hline
$A_4$   & 0000110000   &  0000011000 & 0000100010  \\\hline
$A_5$   & 0000111001   &  0000011011 & 0000101011  \\\hline
$A_6$   & 0000111001   &  0000011000 & 0001000000  \\\hline
$A_7$   & 0000111001   &  0000110000 & 0010000010  \\\hline
\end{tabular}
\end{table}

At the end, we obtain seven type-4 decoding observables, which are listed in Table~\ref{table20}. The form of those observables is given by ${\cal S}^{V}{\cal S}^{(V_1,V_2)}$, which is described in Eq.~(\ref{eqAtype4}). To decide which observable must be measured, we have to analyze Fig.~\ref{table30}. Note that it is enough to measure one non-Pauli observable for this code. We have not put the result ++++ in the list of subtables, because it is trivial --- only the identity operator appears in this case.

\begin{figure}[h!]
  \centering
  \caption{ {Results of the measurements of the decoding observables. The signs on the top of each subtable describe the results of measuring Pauli observables of Table~\ref{table10}. The measurement of non-Pauli observables is conditioned by the results of measuring Pauli observables.}}\label{table30}

\
{
\begin{center}
$\begin{array}{|c|cc|}
 \multicolumn{3}{c}{+++-} \\\hline
    & Y_2 & Z_1\\\hline
A_1 & -   &  +  \\\hline
\end{array}$\,\,\,\,\,
$\begin{array}{|c|cc|}
 \multicolumn{3}{c}{++-+} \\\hline
    & Y_{10} & Z_2\\\hline
A_1 & -   &  +  \\\hline
\end{array}$\,\,\,\,\,
$\begin{array}{|c|cc|}
 \multicolumn{3}{c}{++--} \\\hline
    & X_2 & Z_8\\\hline
A_1 & -   &  +  \\\hline
\end{array}$\,\,\,\,\,
$\begin{array}{|c|cc|}
 \multicolumn{3}{c}{----} \\\hline
    & X_7 & Y_5\\\hline
A_1 & +   &  -  \\\hline
\end{array}$\,\,\,\,\,
$\begin{array}{|c|cc|}
 \multicolumn{3}{c}{+-++} \\\hline
    & X_4 & Z_3\\\hline
A_2 & -   &  +  \\\hline
\end{array}$

\vspace{0.2cm}

$\begin{array}{|c|cc|}
 \multicolumn{3}{c}{---+} \\\hline
    & X_{10} & Z_6\\\hline
A_2 & +   &  -  \\\hline
\end{array}$\,\,\,\,\,
$\begin{array}{|c|cc|}
 \multicolumn{3}{c}{-++-} \\\hline
    & X_3 & Y_7\\\hline
A_3 & +   &  -  \\\hline
\end{array}$\,\,\,\,\,
$\begin{array}{|c|cc|}
 \multicolumn{3}{c}{--+-} \\\hline
    & Y_9 & Y_3\\\hline
A_3 & -   &  +  \\\hline
\end{array}$\,\,\,\,\,
$\begin{array}{|c|cc|}
 \multicolumn{3}{c}{+-+-} \\\hline
    & X_5 & Y_6\\\hline
A_4 & +   &  -  \\\hline
\end{array}$\,\,\,\,\,
$\begin{array}{|c|cc|}
 \multicolumn{3}{c}{--++} \\\hline
    & Z_{10} & Y_4\\\hline
A_4 & +   &  -  \\\hline
\end{array}$

\vspace{0.2cm}

$\begin{array}{|c|cc|}
 \multicolumn{3}{c}{-+++} \\\hline
    & X_8 & Z_4\\\hline
A_5 & -   &  +  \\\hline
\end{array}$\,\,\,\,\,
$\begin{array}{|c|cc|}
 \multicolumn{3}{c}{-+--} \\\hline
    & X_6 & Y_8\\\hline
A_5 & +   &  -  \\\hline
\end{array}$\,\,\,\,\,
$\begin{array}{|c|cc|}
 \multicolumn{3}{c}{-+-+} \\\hline
    & Z_9 & Z_5\\\hline
A_6 & +   &  -  \\\hline
\end{array}$\,\,\,\,\,
$\begin{array}{|c|cc|}
 \multicolumn{3}{c}{+--+} \\\hline
    & X_1 & Z_7\\\hline
A_7 & +   &  -  \\\hline
\end{array}$\,\,\,\,\,
$\begin{array}{|c|cc|}
 \multicolumn{3}{c}{+---} \\\hline
    & X_9 & Y_1\\\hline
A_7 & -   &  +  \\\hline
\end{array}$
\end{center}
}
\end{figure}

\section{Conclusions}

In this work, we have established two results on the existence and form of type-4 decoding observables for CWS codes, namely, Theo.~\ref{Principal} and Corollary~\ref{corollary1}. Those non-Pauli observables are necessary in non-stabilizer CWS codes. We have described a procedure to obtain those observables, which has better chances to succeed when the CWS code is close to a stabilizer code. The standard procedure is to start measuring a list of Pauli observables that stabilizes the code. In the next step, we search for type-4 decoding observables in the search space described by Corollary~\ref{corollary1}.

The procedure does not succeed for all CWS codes, and it is interesting to understand why it fails for some of them. In those cases, is it possible to use type-$i$ observables, with $i>4$ as decoding observables? For example, the $((10,18,3))$ code described in Ref.~\cite{Cross2009} cannot be decoded by type-4 observables.

It is also interesting to study methods, perhaps in family of codes, to obtain the non-Pauli observables in a straightforward way, with less exhaustive search by reducing the search space and to compare with the general method proposed in Ref.~\cite{Li2010}.

\section*{Acknowledgements}

We acknowledge CNPq's financial support

\bibliographystyle{ieeetr}
\bibliography{quantum-codes}

\begin{thebibliography}{10}

\bibitem{Shor1994}
P.~Shor, ``Algorithms for quantum computation: discrete logarithms and
  factoring,'' in {\em Foundations of Computer Science, 1994 Proceedings., 35th
  Annual Symposium on}, pp.~124 --134, nov 1994.

\bibitem{Grover1996}
L.~K. Grover, ``A fast quantum mechanical algorithm for database search,'' in
  {\em Proceedings of the twenty-eighth annual ACM symposium on Theory of
  computing}, STOC '96, (New York, NY, USA), pp.~212--219, ACM, 1996.

\bibitem{Mosca2009}
M.~Mosca, ``Quantum algorithms,'' in {\em Encyclopedia of Complexity and
  Systems Science}, pp.~7088--7118, 2009.

\bibitem{Childs2010}
A.~M. Childs and W.~van Dam, ``Quantum algorithms for algebraic problems,''
  {\em Rev. Mod. Phys.}, vol.~82, pp.~1--52, Jan 2010.

\bibitem{Calderbank1996}
A.~R. Calderbank and P.~W. Shor, ``Good quantum error-correcting codes exist,''
  {\em Phys. Rev. A}, vol.~54, pp.~1098--1105, Aug 1996.

\bibitem{Bennett1996}
C.~H. Bennett, D.~P. DiVincenzo, J.~A. Smolin, and W.~K. Wootters,
  ``Mixed-state entanglement and quantum error correction,'' {\em Phys. Rev.
  A}, vol.~54, pp.~3824--3851, Nov 1996.

\bibitem{Steane1996}
A.~M. Steane, ``Simple quantum error-correcting codes,'' {\em Phys. Rev. A},
  vol.~54, pp.~4741--4751, Dec 1996.

\bibitem{Knill1997}
E.~Knill and R.~Laflamme, ``Theory of quantum error-correcting codes,'' {\em
  Phys. Rev. A}, vol.~55, pp.~900--911, Feb 1997.

\bibitem{Gottesman1996}
D.~Gottesman, ``Class of quantum error-correcting codes saturating the quantum
  hamming bound,'' {\em Phys. Rev. A}, vol.~54, pp.~1862--1868, Sep 1996.

\bibitem{Gottesman1997}
D.~{Gottesman}, {\em {Stabilizer codes and quantum error correction}}.
\newblock PhD thesis, California Institute of Technology, 1997.

\bibitem{Calderbank1997}
A.~R. Calderbank, E.~M. Rains, P.~W. Shor, and N.~J.~A. Sloane, ``Quantum error
  correction and orthogonal geometry,'' {\em Phys. Rev. Lett.}, vol.~78,
  pp.~405--408, Jan 1997.

\bibitem{Smolin2007}
J.~A. Smolin, G.~Smith, and S.~Wehner, ``Simple family of nonadditive quantum
  codes,'' {\em Phys. Rev. Lett.}, vol.~99, p.~130505, Sep 2007.

\bibitem{Yu2008}
S.~Yu, Q.~Chen, C.~H. Lai, and C.~H. Oh, ``Nonadditive quantum error-correcting
  code,'' {\em Phys. Rev. Lett.}, vol.~101, p.~090501, Aug 2008.

\bibitem{Cross2008}
A.~Cross, G.~Smith, J.~Smolin, and B.~Zeng, ``Codeword stabilized quantum
  codes,'' in {\em Information Theory, 2008. ISIT 2008. IEEE International
  Symposium on}, pp.~364 --368, july 2008.

\bibitem{Chen2008}
X.~Chen, B.~Zeng, and I.~L. Chuang, ``Nonbinary codeword-stabilized quantum
  codes,'' {\em Phys. Rev. A}, vol.~78, p.~062315, Dec 2008.

\bibitem{Chuang2009}
I.~Chuang, A.~Cross, G.~Smith, J.~Smolin, and B.~Zeng, ``Codeword stabilized
  quantum codes: Algorithm and structure,'' {\em Journal of Mathematical
  Physics}, vol.~50, no.~4, p.~042109, 2009.

\bibitem{Yu2009}
S.~{Yu}, Q.~{Chen}, and C.~H. {Oh}, ``{Two infinite families of nonadditive
  quantum error-correcting codes},'' {\em ArXiv e-prints}, Jan. 2009.

\bibitem{Li2010}
Y.~Li, I.~Dumer, M.~Grassl, and L.~P. Pryadko, ``Structured error recovery for
  code-word-stabilized quantum codes,'' {\em Phys. Rev. A}, vol.~81, p.~052337,
  May 2010.

\bibitem{Melo2012}
N.~{Melo}, D.~F.~G. {Santiago}, and R.~{Portugal}, ``{Decoder for Nonbinary CWS
  Quantum Codes},'' {\em ArXiv e-prints}, Apr. 2012.

\bibitem{Cross2009}
A.~Cross, G.~Smith, J.~A. Smolin, and B.~Zeng, ``Codeword stabilized quantum
  codes,'' {\em IEEE Trans. Inf. Theor.}, vol.~55, pp.~433--438, Jan. 2009.

\end{thebibliography}

\end{document}